\documentclass[letterpaper,11pt]{article}
\usepackage[utf8]{inputenc}

\pdfoutput=1

\usepackage{amsmath, amsthm, amssymb, thm-restate}
\usepackage{algorithmicx}
\usepackage[table,xcdraw]{xcolor}
\usepackage[ruled,vlined,linesnumbered]{algorithm2e}
\usepackage{setspace}
\usepackage{mathtools}
\usepackage[numbers]{natbib}
\usepackage{comment} 
\usepackage{tcolorbox} 
\usepackage{xfrac}
\usepackage{hyperref}
\usepackage[noabbrev,nameinlink]{cleveref}
\usepackage{multirow}
\usepackage{caption}
\usepackage{bm}
\usepackage{newfloat}
\usepackage{enumitem}
\usepackage{dblfloatfix} 
\usepackage{wrapfig}
\usepackage{csquotes}

\usepackage{fancyhdr}

\usepackage[margin=1in]{geometry}

\allowdisplaybreaks

\definecolor{mygreen}{RGB}{10,150,110}
\definecolor{myred}{RGB}{150,10,20}

\hypersetup{
     colorlinks=true,
     citecolor= mygreen,
     linkcolor= myred
}

\renewcommand{\epsilon}{\varepsilon}

\newcommand{\hiddencomment}[1]{}

\newcommand{\ep}{\mathsf{e}}
\newcommand{\vp}{\mathsf{v}}

\newcounter{lpcounter}
\newenvironment{lp}{
\refstepcounter{lpcounter}
\begin{centering}
\begin{halfwhitetbox}
\textbf{LP \thelpcounter:}%
}{
\end{halfwhitetbox}
\end{centering}} 
\crefname{lpcounter}{LP}{LPs}

\newcommand{\card}[1]{\left\lvert#1\right\rvert}

\crefname{lemma}{Lemma}{Lemmas}
\crefname{theorem}{Theorem}{Theorems}
\crefname{property}{Property}{Properties}
\crefname{claim}{Claim}{Claims}
\crefname{result}{Result}{Results}
\crefname{definition}{Definition}{Definitions}
\crefname{observation}{Observation}{Observations}
\crefname{proposition}{Proposition}{Propositions}
\crefname{assumption}{Assumption}{Assumptions}
\crefname{line}{Line}{Lines}
\crefname{figure}{Figure}{Figures}
\creflabelformat{property}{(#1)#2#3}
\crefname{equation}{}{}
\crefname{section}{Section}{Sections}
\crefname{appendix}{Appendix}{Appendices}
\crefname{algCounter}{Algorithm}{Algorithms}
\Crefname{algCounter}{Algorithm}{Algorithms}

\newtheorem{theorem}{Theorem}

\newtheorem{lemma}{Lemma}[section]
\newtheorem{proposition}[lemma]{Proposition}

\newtheorem{definition}[lemma]{Definition}
\newtheorem{claim}[lemma]{Claim}

\newtheorem*{remark*}{Remark}

\definecolor{mylightgray}{RGB}{230,230,230}

\algnewcommand{\IIf}[2]{\textbf{if} #1 \textbf{then} #2}
\algnewcommand{\EndIIf}{\unskip\ \algorithmicend\ \algorithmicif}

\newenvironment{whitetbox}{
\par\addvspace{0.1cm}
\begin{tcolorbox}[width=\textwidth,
                  boxsep=5pt,
                  left=1pt,
                  right=1pt,
                  top=2pt,
                  bottom=2pt,
                  boxrule=1pt,
                  arc=0pt,
                  colframe=black,
                  colback=white
                  ]
}{
\end{tcolorbox}
}

\newcounter{algCounter}

\newenvironment{halfwhitetbox}{
\par\addvspace{0.1cm}
\begin{tcolorbox}[width=0.7\columnwidth,
                  boxsep=5pt,
                  left=1pt,
                  right=1pt,
                  top=2pt,
                  bottom=2pt,
                  boxrule=1pt,
                  arc=0pt,
                  colframe=black,
                  colback=white
                  ]
}{
\end{tcolorbox}
}

\makeatletter
\renewcommand{\paragraph}{%
  \@startsection{paragraph}{4}%
  {\z@}{10pt}{-1em}%
  {\normalfont\normalsize\bfseries}%
}
\makeatother

\makeatletter
\patchcmd{\@algocf@start}
  {-1.5em}
  {0pt}
  {}{}
\makeatother

\title{Bipartite Matching in Massive Graphs:\\ A Tight Analysis of EDCS}

\author{
Amir Azarmehr\\{\em Northeastern University} \and 
Soheil Behnezhad \\{\em Northeastern University} \and
Mohammad Roghani \\{\em Stanford University}
}

\date{}

\begin{document}

\maketitle

\thispagestyle{empty}
\begin{abstract}
Maximum matching is one of the most fundamental combinatorial optimization problems with applications in various contexts such as balanced clustering, data mining, resource allocation, and online advertisement. In many of these applications, the input graph is massive. The sheer size of these inputs makes it impossible to store the whole graph in the memory of a single machine and process it there. Graph sparsification has been an extremely powerful tool to alleviate this problem. In this paper, we study a highly successful and versatile sparsifier for the matching problem: the {\em edge-degree constrained subgraph (EDCS)} introduced first by Bernstein and Stein [ICALP'15].

\smallskip\smallskip
The EDCS has a parameter $\beta \geq 2$ which controls the density of the sparsifier. It has been shown through various proofs in the literature that by picking a subgraph with $O(n\beta)$ edges, the EDCS includes a matching of size at least $2/3-O(1/\beta)$ times the maximum matching size. As such, by increasing $\beta$ the approximation ratio of EDCS gets closer and closer to $2/3$.

\smallskip\smallskip
In this paper, we propose a new approach for analyzing the approximation ratio of EDCS. Our analysis is {\em tight} for any value of $\beta$. Namely, we pinpoint the precise approximation ratio of EDCS for any sparsity parameter $\beta$. Our analysis reveals that one does not necessarily need to increase $\beta$ to improve approximation, as suggested by previous analysis. In particular, the best choice turns out to be $\beta = 6$, which achieves an approximation ratio of $.677$! This is arguably surprising as it is even better than $2/3 \sim .666$, the bound that was widely believed to be the limit for EDCS.
\end{abstract}

{
\clearpage
\hypersetup{hidelinks}
\vspace{1cm}
\renewcommand{\baselinestretch}{0.1}
\setcounter{tocdepth}{2}
\thispagestyle{empty}
\clearpage
}

\setcounter{page}{1}
\section{Introduction} \label{sec:introduction}

Maximum matching is one of the most fundamental combinatorial optimization problems. Recall that a {\em matching} in a graph is a collection of edges that do not share any vertices. A {\em maximum} matching is a matching of the largest possible size. 

The matching problem finds applications in various contexts such as data mining, resource allocation, online advertisement, bioinformatics, and many others. For instance, maximum matching can improve the quality of data clustering \cite{AssadiBM19}, it can produce fair  $k$-center clustering \cite{JonesNN20}, or can be used to discover subgraphs for bioinformatics applications \cite{BergerSX08,LangmeadD04}. In most of these applications, the input graph is massive. The sheer size of these inputs makes it impossible to store the whole graph in the memory of a single machine and process it there.  This has motivated a large and beautiful body of work over the past two decades on large-scale algorithms for this problem.

\paragraph{Graph Sparsification:} Graph sparsification is a powerful tool to process massive graphs. A graph sparsifier receives an $n$-vertex graph that may have as many as $\Omega(n^2)$ edges and sparsifies it into a sparse subgraph, say with $O(n)$ edges, that preserves some property of it. Graph sparsifiers have been instrumental tools for various graph problems. Cut sparsifiers \cite{NagamochiI92}, spectral sparsifiers \cite{SpielmanT11}, and spanners \cite{AbboudB17} are some famous examples. Graph sparsifiers did not find many applications for matchings until nearly a decade ago when \citet{BernsteinS15} introduced the {\em edge-degree constrained subgraph (EDCS)}. See in particular the nice paper of \citet{AssadiB19} for an in-depth introduction to EDCS and an overview of some of its applications.

Over the years, the EDCS has been successfully applied to a variety of large-scale settings including the massively parallel computations (MPC) setting which is a common theoretical model of MapReduce-style computation \cite{AssadiBBMS19}, the dynamic setting \cite{BernsteinS15, BehnezhadK22, RSW22}, the streaming setting \cite{Bernstein20,AssadiB21}, the sublinear time setting \cite{BehnezhadRR23,BhattacharyaKS23}, communication complexity \cite{azarmehrBehnezhad},  and the stochastic matching setting \cite{AssadiB19}.

In this paper, we revisit the key property of EDCS: that it obtains a good approximation of maximum matching while at the same time being sparse. To put our results into perspective, we first need to provide some background and overview existing bounds.

\subsection*{Background}

Let us start by stating the formal definition of EDCS. 

\begin{definition}[\citet{BernsteinS15}]\label{def:EDCS}
    \label{def:edcs}
    Given a graph $G$, a subgraph $H \subseteq G$ is an {\em edge-degree constrained subgraph} with parameters $(\beta, \beta^-)$, or a $(\beta, \beta^-)$-EDCS,
    if the following conditions hold:
    \begin{enumerate}
        \item for all edges $(u, v) \in H$,  $\deg_H(u) + \deg_H(v) \leq \beta$, and
    \item for all edges $(u, v) \in G \setminus H$,  $\deg_H(u) + \deg_H(v) \geq \beta^-$.
    \end{enumerate}
\end{definition}

The following proposition shows that a $(\beta, \beta^-)$-EDCS always exists for all integers $\beta > \beta^- \geq 1$.

\begin{proposition}[\citet{BernsteinS15}]\label{prop:existence}
Any graph $G$ contains a $(\beta, \beta^-)$-EDCS
for any integers $\beta > \beta^- \geq 1$,
and one can be found greedily in polynomial time.
\end{proposition}

Moreover, since the edge-degrees in a $(\beta, \beta^-)$-EDCS are all upper bounded by $\beta$ by the first property of \cref{def:EDCS}, so are the vertex degrees. Therefore, the EDCS has at most $O(n\beta)$ edges. This means that smaller values of $\beta$ are more desirable as the subgraph picked will be a sparser.

It will be instructive to set $\beta = 2$ and $\beta^- = 1$. It can be easily confirmed that any $(2, 1)$-EDCS is a {\em maximal matching} (i.e., a matching that is not a subset of another matching) and that any maximal matching is a $(2, 1)$-EDCS. It is well-known that a maximal matching has at least half as many edges as a maximum matching and that this bound is tight.\footnote{Take $G$ to be a path with 3 edges and take $H$ to be the subgraph only containing the middle edge.} The key property of EDCS is that by slightly increasing $\beta$ and keeping $\beta^-$ close to it, the approximation ratio improves to almost 2/3. Formally:

\begin{proposition}[\citet{BernsteinS15,BernsteinS16,AssadiB19,corr/Behnezhad21}]\label{prop:key}
Given a graph $G$, and parameters $\beta \geq 1/\epsilon$ and $\beta^-\geq (1 - \epsilon)\beta$, any $(\beta, \beta^-)$-EDCS of $G$ contains a $(2/3 - O(\epsilon))$-approximate maximum matching of $G$.
\end{proposition}

The approximation guarantee of \cref{prop:key} is close to optimal. In particular, for infinitely many choices of $$1 \leq \beta^- < \beta,$$ (particularly for all odd $\beta = 2k+1$ and all $\beta^- < \beta$) examples have been known since the original paper of \cite{BernsteinS15} where a $(\beta, \beta^-)$-EDCS does not include a better than 2/3-approximation. See \cref{fig:tight}. 

\begin{figure}
    \centering
    \includegraphics[width=0.7\textwidth]{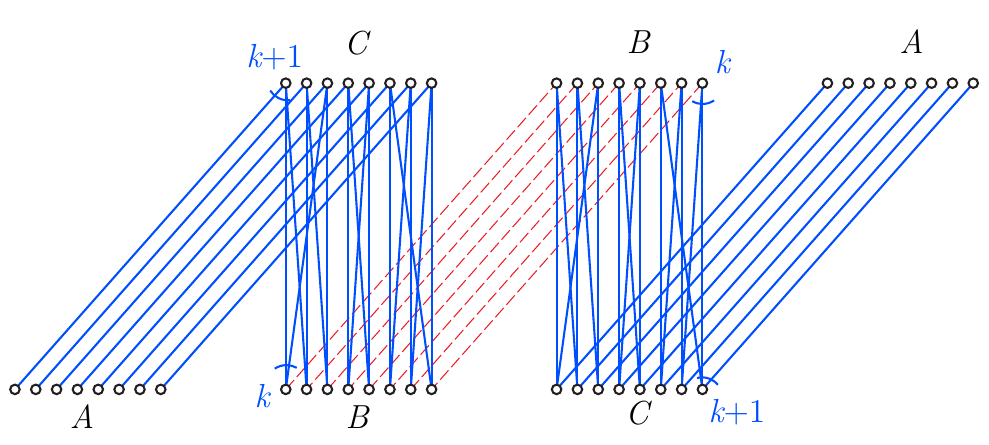}
    \caption{An example due to \cite{BernsteinS15} where a $(\beta, \beta-1)$-EDCS does not obtain a better than 2/3-approximation for any odd $\beta = 2k+1$. Here we have a bipartite graph, each side having three equally sized vertex sets $A, B, C$. The solid blue lines denote the EDCS edges, and the dashed red lines denote the edges not in the EDCS. The vertices in $B$ have degree $k$ in the EDCS, the vertices in $C$ have degree $k+1$ in the EDCS, and the vertices in $A$ have degree one. Note that any edge in the EDCS has edge degree at most $\beta = 2k+1$. The only edges missing from the EDCS are the dashed $B$-$B$ edges which all have edge degree exactly $\beta - 1 = 2k$. While the graph $G$ has a perfect matching, the EDCS can only match $2/3$ of the vertices.}
    \label{fig:tight}
\end{figure}

\subsection*{Our Contribution}

In most applications of EDCS, we would like to set $\beta$ to be as small as possible to achieve sparser subgraphs.\footnote{The only exception is the deterministic dynamic algorithm of \cite{BernsteinS15} for maintaining EDCS where larger values of $\beta$ help as they make the EDCS more ``robust'' to changes.} On the other hand, making $\beta$ smaller would make the approximation guarantee of \cref{prop:key} worse. It is therefore natural to study the trade-off between $\beta$ and the approximation ratio. Unfortunately, known proofs are too loose, especially, for small values of $\beta$. In this paper, we propose a new approach to analyze the approximation ratio of EDCS. Our analysis is tight and precisely pinpoints the exact approximation ratio achieved for any given $\beta$. \cref{tab:small-values} states the approximation ratio for various values of $\beta$ and $\beta^-$.

\paragraph{Our Results:} Our analysis reveals that the approximation ratio of $(\beta, \beta^-)$-EDCS, when $\beta^- = \beta-1$, behaves very differently from when $\beta^- < \beta - 1$. Note that this is still in the regime where a $(\beta, \beta^-)$-EDCS exists and can be found in polynomial time. For instance, for any odd value of $\beta \geq 7$, a $(\beta, \beta-1)$-EDCS obtains an exact $2/3$-approximation. The approximation turns out to be quite surprising when considering even values of $\beta \geq 6$. For instance, a $(6, 5)$-EDCS obtains a $.677$-approximation which is even better than $2/3 \sim .666$! While this may seem to contradict the example of \cref{fig:tight}, it has to be noted that \cref{fig:tight} requires $\beta$ to be odd and does not work when it is even. Moreover, as we increase $\beta$ in the even case, the approximation ratio gets worse and approaches $2/3$ (see \cref{fig:beta-minus-one}). We note that when $\beta$ is even, $\beta-1$, which is the edge-degree lower bound for missed edges, is odd. This means that the two endpoints of such missed edges must have different degrees. Such imbalance between the degrees of the vertices with missed optimal edges is precisely the reason for a better than 2/3 approximation. When we increase $\beta$, this imbalance becomes less significant (as the ratio of degrees gets closer and closer to 1) and so the approximation becomes worse. This is in sharp contrast with previous analysis such as the one in \cref{prop:key} where increasing $\beta$ improves the approximation.

We emphasize that going beyond 2/3-approximation for the maximum matching problem is considered a difficult task in many settings. We refer the interested reader to \cite{AssadiB21,BehnezhadRR23} where the precise problem of beating 2/3-approximation is studied in various settings. We hope that our discovery that a $(6, 5)$-EDCS beats 2/3-approximation combined with the known bound of \cref{prop:existence} that such EDCS's can be found via a simple greedy algorithm paves the way for future progress on this important question.

\paragraph{Our Analysis:} While previous analysis of EDCS were analytical, our analysis is based on a new factor-revealing linear program (formalized as \cref{lp:main}) that we show provides the exact approximation ratio of $(\beta, \beta^-)$-EDCS for any given parameters $\beta, \beta^-$. 

We then provide the claimed approximation guarantees by solving this LP. We note that a factor revealing LP has also been used to analyze a hierarchical version of EDCS in \cite{BehnezhadK22}. However, the factor revealing LP there is different from ours and, importantly, is not tight. For instance, the LP used by \cite{BehnezhadK22} only guarantees a $0.6$-approximation for a $(6, 5)$-EDCS which is way smaller than the correct bound of $0.677$ returned by our tight LP.
\section{Preliminaries} \label{sec:preliminaries}
In this section, we introduce the notations and definitions we use, and provide some background on matchings.

A graph $G$ is \emph{bipartite} if its vertices can be partitioned into two sets $L$ and $R$, such that every edge has exactly one endpoint in $L$ and one endpoint in $R$. We use $G(L, R)$ to denote a bipartite graph with partitions $L$ and $R$.

Given a set of vertices $A$, we use $N_G(A)$ (or $N(A)$ when $G$ is clear from the context) to denote the set of its neighbors. For a vertex $u$, we use $\deg_G(u)$ to denote its degree in $G$.

Given a graph $G$, a \emph{matching} is a subset of edges such that no two edges share an endpoint. We say that a matching $M$ covers a vertex $u$, or that $u$ is matched in $M$, if there is an edge adjacent to $u$ in $M$. A \emph{maximum matching} is a matching with the largest possible number of edges. We use $\mu(G)$ to denote the size of a maximum matching in $G$.

The following is a well-known primal-dual result for bipartite maximum matching.
\begin{proposition}[Extended Hall's Theorem]
\label{prp:hall}
Given a bipartite graph $G(L, R)$,
it holds that
$$
\mu(G) = \min_{A \subseteq L} \card{N(A)} + \card{L \setminus A}.
$$
The vertex set that minimizes the right-hand side is referred to as a \emph{Hall's witness} for $G$. 
Furthermore, for every maximum matching $M$ and Hall's witness $A$, every edge of $M$ covers exactly one vertex in $N(A) \cup (L \setminus A)$.
\end{proposition}

\begin{figure}[h]
    \centering
    \includegraphics[width=0.7\textwidth]{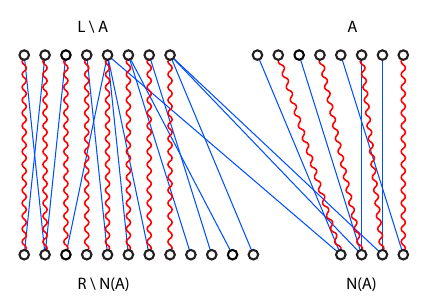}
    \caption{An example of Hall's witness. 
    The curvy red edges denote a maximum matching, and $A$ is a Hall's witness. 
    Generally, considering a vertex set $A \subseteq L$ and any matching, the vertices in $A$ are matched to a subset of $N(A)$. Therefore, even if all the vertices in $L \setminus A$ are somehow matched, the matching has size at most $\card{N(A)} + \card{L\setminus A}$, i.e.\ $\mu(G) \leq \card{N(A)} + \card{L\setminus A}$. Hall's theorem states that there exists a vertex set $A$, referred to as a Hall's witness, for which this inequality is tight.} 
    \label{fig:halls}
\end{figure}

\section{A New Analysis of the EDCS via a Factor-Revealing LP}
\label{sec:lp}
In this section, we present the linear program that \enquote{reveals} the approximation ratio of EDCS for fixed parameters $(\beta, \beta^-)$.
We prove that it is tight, i.e.\ any pair $(G, H)$ consisting of a bipartite graph $G(L, R)$ and a $(\beta, \beta^-)$-EDCS $H \subseteq G$ can be converted to a feasible solution of the LP and vice versa.

The edges and the vertices of the graph are divided into groups based on their properties and their role in the graph (e.g.\ for a vertex this includes its degree, whether it is matched in the maximum matching, where it is in the Hall's witness, etc.) We refer to these properties as a \emph{vertex profile} or an \emph{edge profile}, and use $VP$ and $EP$ to denote the set of valid vertex profiles and edge profiles, respectively.
Then, a variable is designated to each group of vertices or edges with the same profile.
Its value is set (proportionally) to the number of vertices or edges in that group.

The easiest way to understand the LP is to see how a pair $(G, H)$ is converted to a feasible solution, where $G$ is a bipartite graph, and $H$ is a $(\beta, \beta^-)$-EDCS for $G$.
First, fix a maximum matching $M^*$ of $G$ and a maximum matching $M$ of $H$, along with a Hall's witness $A$ for $H$, i.e.\ $\card{M} = \mu(H) = \card{N_H(A)} + \card{L \setminus A}$.
Without loss of generality, we can assume that $G$ contains no edges other than $M^* \cup H$.
Because those edges can be removed from the graph, in which case $\mu(G)$ and $\mu(H)$ remain unchanged while $H$ is still a $(\beta, \beta^-)$-EDCS.

We divide the vertices and the edges into different groups.
All the vertices or edges in a group have the same properties, a.k.a.\ profile, which will be defined shortly.
Each variable of the LP then reflects how many vertices or edges with each profile there are in the graph.
We scale all the numbers by $1/\mu(H)$,
that is if there are $k$ vertices with a certain profile, then the variable corresponding to that vertex profile holds the value $k/\mu(H)$.
As a result, the variables corresponding to the edges in $M$ sum to $1$, 
the variables corresponding to the edges in $M^*$ sum to the approximation ratio $\mu(G) / \mu(H)$ which we set as the objective function of our LP to be \emph{maximized}.

Now, we define the vertex profiles. To do so, we consider all the possible cases of the following properties for a vertex:
\begin{itemize}
    \item whether it is in $A$, $L \setminus A$, $N_H(A)$, or $R \setminus N_H(A)$,
    \item its degree in $H$, an integer between $0$ and $\beta - 1$,
    \item whether it is matched in $M$, and
    \item whether it is matched in $M^*$.
\end{itemize}
Since we create a variable for each vertex profile, we make sure to use only valid vertex profiles, so that we do not create \enquote{extra} variables.
That is, we have to confirm that the properties above in a profile make sense together. 
We have a total of ten validity conditions, two for the vertex profiles and (as explained later) eight for the edge profiles. These conditions are enforced when formulating the LP, and they are not a part of the LP itself.
Specifically, for the vertex profiles, we assert the following:
\begin{enumerate}
    \item \label{cnd:hall1} if a vertex is in $N_H(A)$ or $L \setminus A$, then it must be matched in $M$, because $M$ is a maximum matching of $H$ and $A$ is a Hall's witness for $H$ (see \cref{prp:hall}), and
    \item if a vertex has zero degree in $H$, then it must be unmatched in $M$, since having an edge in $M$ would mean having degree at least $1$.
\end{enumerate}
After considering all the cases and discarding the vertex profiles that do not satisfy the aforementioned conditions, we create a variable for each vertex profile. Observe that any vertex in the graph has a valid profile, i.e.\ its properties satisfy the two conditions. Finally, to set the values in our feasible solution, if there are $n_\vp$ vertices with profile $\vp$, we let the variable $x_\vp$ corresponding to that vertex profile be equal to $n_\vp/\mu(H)$.

To define an edge profile, we consider all the possible cases of the following properties for an edge:
\begin{itemize}
    \item the vertex profiles of its endpoints in $L$ and $R$,
    \item whether it is in $H$,
    \item whether it is in $M$, and
    \item whether it is in $M^*$.
\end{itemize}
Checking the validity of an edge profile takes more work. This is where we enforce the bulk of the properties that $G$, $H$, and $M$ have.
Note that if a profile is invalid, we do not create a variable for it at all.
First, some simple consistency conditions:
\begin{enumerate}[resume]
    \item if the edge is in $H$, then its endpoint vertex profiles must indicate a nonzero degree in $H$,
    \item if the edge is in $M$ (resp.\ $M^*$), its endpoint vertex profiles must indicate that they are matched in $M$ (resp.\ $M^*$),
    \item if an edge is in $M$, it must be (by definition) in $H$, and
    \item there are no edges outside $H \cup M^*$ (see the beginning of \cref{sec:lp}).
\end{enumerate}
Then, the conditions concerning the EDCS $H$ (see \cref{def:edcs}):
\begin{enumerate}[resume] 
    \item \label{cnd:edcs1} if an edge is in $H$, the degrees of its endpoints in $H$ must sum to at most $\beta$, and
    \item \label{cnd:edcs2}if an edge is not in $H$, the degrees of its endpoints in $H$ must sum to at least $\beta^-$.
\end{enumerate}
Finally, the conditions for the Hall's witness $A$ (see \cref{prp:hall}):
\begin{enumerate}[resume]
    \item \label{cnd:hall2} there should be no edges of $H$ between $A$ and $R \setminus N_H(A)$ (by definition), and
    \item \label{cnd:hall3} for every edge of $M$, \emph{exactly} one of the following holds: ($i$) it has an endpoint in $N(A)$, or ($ii$) it has an endpoint in $L \setminus A$.
\end{enumerate}
Similar to the vertices, we create a variable for each valid edge profile (i.e.\ an edge profile that satisfies all the conditions), and set their values proportional to the number of edges with that profile. That is, if there are $n_\ep$ edges with profile $\ep$, the variable $x_\ep$ corresponding to that profile is set equal to $\frac{n_\ep}{\mu(H)}$. This completes the explanation of how the variables are meant to correspond to a graph.

To tie this all together, we need to add the constraints.
Most of the properties of $G$, $H$, and $M$ are already encoded in the vertex/edge profiles.
The purpose of the constraints is to link the number of vertices to the number of edges adjacent to them, and to ensure that the variables corresponding to $M$ sum to $1$, i.e.\ everything is scaled by $1/\mu(H)$.

Recall, $VP$ and $EP$ denote the set of valid vertex profiles and edge profiles, respectively.
For every profile $\vp \in VP$ (resp.\ $\ep \in EP$),
we create a variable, and denote it by $x_{\vp}$ (resp.\ $x_{\ep}$).
We use $H(\vp)$ (resp.\ $M(\vp)$, $M^*(\vp)$)
to denote the set of edge profiles that are in $H$ (resp.\ $M$, $M^*$) and one of their endpoints has profile $\vp$.
We also use $\ep \in M^*$ to denote that the edges with profile $\ep$ are in $M^*$, and $\vp_{\deg}$ to denote the degree of vertices with profile $\vp$ in $H$. The LP is then as follows:
\begin{lp}
The factor-revealing LP for the approximation ratio of $(\beta, \beta^-)$-EDCS
\label{lp:main}
\begin{flalign*}
    \text{maximize } \sum\limits_{\substack{\ep \in EP \\ \ep \in M^*}} x_{\ep} & \\
    \sum_{\ep \in H(\vp)} x_{\ep} &= \vp_{\deg} \cdot x_{\vp} & \forall \vp \in VP \\
    \sum_{\ep \in M(\vp)} x_{\ep} &= x_{\vp} & \substack{\forall \vp \in VP \\ \text{ that is matched in $M$}} \\
    \sum_{\ep \in M^*(\vp)} x_{\ep} &= x_{\vp} & \substack{\forall \vp \in VP \\ \text{ that is matched in $M^*$}} \\
    \sum_{\ep \in M} x_\ep &= 1 \\
    x_{\ep} &\geq 0 &\forall \ep \in EP \\
    x_{\vp} &\geq 0 &\forall \vp \in VP \\
\end{flalign*}
\end{lp}
To see how the first three constraints tie the vertex variables to edge variables, take the first set of constraints as an example.
These constraints state that for a vertex profile that has degree $d$ in $H$, if there are $k$ vertices with this profile, then there should be $d \cdot k$ edges of $H$ adjacent to these vertices. The fourth constraint simply states that everything is scaled by $1/\mu(H)$, therefore the variables corresponding to the edges of $M$ sum to $1$ (recall $\card{M} = \mu(H)$). With that, we are ready to state the main theorem of this section.

\begin{theorem}
    \label{thm:main}
    The optimal objective value for \cref{lp:main} is equal to the approximation ratio of $(\beta, \beta^-)$-EDCS. That is, if the optimal value for \cref{lp:main} is $r$, then for every bipartite graph $G$ and every $(\beta, \beta^-)$-EDCS $H \subseteq G$, it holds that $r \cdot \mu(H) \geq \mu(G)$.
    Furthermore, there exists an instance where this inequality is tight. 
\end{theorem}

We prove the theorem by showing the following two claims hold.

\begin{claim}
    The optimal value of \cref{lp:main}
    is an upper bound for the approximation ratio of $(\beta, \beta^-)$-EDCS.
    That is, any pair of a bipartite graph $G$ and $(\beta, \beta^-)$-EDCS $H \subseteq G$
    can be converted to a feasible solution of $\cref{lp:main}$ such that the objective value is equal to $\mu(G) / \mu(H)$.
\end{claim}
\begin{proof}
    The reduction has been partially explained.
    Given $G$ and $H$, we fix a maximum matching $M^*$ in $G$,
    a maximum matching $M$ in $H$,
    and a Hall's witness $A$ for $H$.
    Then for every vertex profile $\vp$ (similarly for every edge profile $\ep$), if the number of vertices with that profile is $n_{\vp}$, we set the value of the corresponding variable $x_{\vp}$ equal to $n_{\vp}/\mu(H)$.
    Note that any vertex (similarly edge) of $G$ corresponds to exactly one vertex profile and the constraints hold automatically. Therefore, $x$ is feasible.

    Now we calculate the objective value for $x$ (denote it by $r$).
    $$
    r = \sum_{\ep\in M^*} x_{\ep}
    = \frac{\sum_{\ep\in M^*} n_{\ep}}{\mu(H)} = \frac{\mu(G)}{\mu(H)},
    $$
    which concludes the proof.
\end{proof}

\begin{claim}
    The optimal value of \cref{lp:main}
    is a lower bound for the approximation ratio of $(\beta, \beta^-)$-EDCS.
    That is, an optimal solution of \cref{lp:main} with objective value $r$ can be converted to a bipartite graph $G$ and a $(\beta, \beta^-)$-EDCS $H \subseteq G$ such that $\mu(G) / \mu(H) = r$.
\end{claim}
\begin{proof}
    Take a rational optimal solution $x$ (note that since the coefficients in the constraints are rational, there exists a rational optimal solution to the LP).
    Because $x$ is rational, there exists a positive integer $N$ such that $n_{\vp} = N \cdot x_{\vp}$ and $n_{\ep} = N \cdot x_{\ep}$ is an integer for all profile $\vp$ and $\ep$.

    We create a graph $G$ with exactly $n_{\vp}$ vertices with profile $\vp$, and $n_{\ep}$ edges with profile $\ep$.
    To start, we create a group of $n_{\vp}$ vertices for each vertex profile $\vp$.
    To connect them, for each edge profile $\ep$ we use $n_{\ep}$ edges between the two endpoint groups of $\ep$ (recall that each edge profile indicates the vertex profiles of its endpoints). This can be done because of the LP constraints that link the number of edges to the number of vertices.
    More specifically, one can start with the edge profiles of $M^*$ and match unmatched vertices from the endpoint groups.
    Then, move on to the edge profiles of $M \setminus M^*$
    and match unmatched vertices (w.r.t.\ $M$) from the endpoint groups.
    Finally, go over the edge profiles of $H \setminus (M \cup M^*)$
    and connect the vertices from the endpoint groups to achieve their designated degree in $H$.

    Now that we have $G$, we define $H$, $M$, $M^*$, and $A$ simply by considering what the vertex/edge profiles indicate.
    $H$ is a $(\beta, \beta^-)$-EDCS of $G$
    because of the EDCS validity conditions for edge profiles (conditions \ref{cnd:edcs1} and \ref{cnd:edcs2}).
    $M$ is a maximum matching of $H$ since it is coupled with the Hall's witness $A$.
    That is, there is a vertex set $A \subseteq L$ such that each edge of $M$ matches exactly one vertex from $(L \setminus A) \cup N(A)$ (conditions \ref{cnd:hall1}, \ref{cnd:hall2}, and \ref{cnd:hall3}).
    Finally, $M^*$ is a maximum matching since otherwise we could choose a larger matching as $M^*$ and derive a feasible solution with a larger objective value, which contradicts optimality.

    Now we calculate the approximation ratio in this instance:
    $$
    \frac{\mu(G)}{\mu(H)}
    = \frac{\sum_{\ep\in M^*} n_{\ep}}{\mu(H)}
    = \sum_{\ep\in M^*} x_{\ep}.
    $$
    The right-hand side is the objective value for $x$, i.e.\ the optimal value of \cref{lp:main}, which concludes the proof.
\end{proof}

Putting the two claims together, gives \cref{thm:main}.
\section{Numerical Solutions}
\label{sec:numerical}

\subsection{Setup}
All of our code\footnote{The implemented code can be found at the following \href{https://github.com/mohammadroghani/Tight-Analysis-of-EDCS}{link}.} is written in Python (version 3.10.12) and is available in the supplementary material. For solving factor-revealing LP instances, we utilized the Gurobi optimization package (version 11.0.0). The experiments were conducted on a computing cluster equipped with 64 cores, each running at 2.30GHz on Intel(R) Xeon(R) processors, and with 756 GiB of main memory. The operating system used was Ubuntu 22.04.3 LTS.

\subsection{Results}
We implement the factor-revealing LP in \Cref{sec:lp} for all possible values of $\beta, \beta^{-} \in [1, 100]$ where $\beta > \beta^{-}$. The largest approximation ratio that we obtain using the factor-revealing LP is 0.6774 which is achieved for parameter $\beta = 6$ and $\beta^{-}=5$ (see \Cref{tab:small-values} and \Cref{fig:heatmap}).

\begin{table*}[t]
  \centering
  \resizebox{\columnwidth}{!}{%
      \begin{tabular}{|c||c|c|c|c|c|c|c|c|c|c|c|c|}
      \hline
         $\beta \backslash \beta^{-}$& 1 & 2 & 3 & 4 & 5 & 6 & 7 & 8 & 9 & 10 & 11\\
        \hline
    2 & 0.5  & -  & -  & -  & -  & -  & -  & -  & -  & -  & - \\
    \hline
    3 & 0.3333  & 0.5  & -  & -  & -  & -  & -  & -  & -  & -  & - \\
    \hline
    4 & 0.25  & 0.4  & 0.625  & -  & -  & -  & -  & -  & -  & -  & - \\
    \hline
    5 & 0.2  & 0.3333  & 0.4782  & 0.6249  & -  & -  & -  & -  & -  & -  & - \\
    \hline
    6 & 0.1666  & 0.2857  & 0.4117  & 0.5  & 0.6774  & -  & -  & -  & -  & -  & - \\
    \hline
    7 & 0.1428  & 0.25  & 0.3617  & 0.4444  & 0.5604  & 0.6666  & -  & -  & -  & -  & - \\
    \hline
    8 & 0.125  & 0.2222  & 0.3225  & 0.4  & 0.4827  & 0.5783  & 0.6756  & -  & -  & -  & - \\
    \hline
    9 & 0.1111  & 0.2  & 0.2911  & 0.3636  & 0.4399  & 0.5  & 0.6097  & 0.6666  & -  & -  & - \\
    \hline
    10 & 0.1  & 0.1818  & 0.2653  & 0.3333  & 0.4042  & 0.4615  & 0.539  & 0.6153  & 0.6721  & -  & - \\
    \hline
    11 & 0.0909  & 0.1666  & 0.2436  & 0.3076  & 0.3739  & 0.4285  & 0.4862  & 0.5569  & 0.625  & 0.6666  & - \\
    \hline
    12 & 0.0833  & 0.1538  & 0.2253  & 0.2857  & 0.3478  & 0.3999  & 0.4545  & 0.5  & 0.5796  & 0.625  & 0.6703 \\
    \hline
    
      \end{tabular}
  }
  \caption{The approximation ratio achieved by EDCS obtained from the factor-revealing LP in \Cref{sec:lp}. Rows correspond to $\beta$, and columns correspond to $\beta^{-}$. The table presents the approximation ratio for all possible $(\beta, \beta^{-})$-EDCS when $1 \leq \beta^{-} < \beta \leq 12$. A hyphen in a cell of the table indicates that there is no EDCS with the corresponding parameters.}
  \label{tab:small-values}
\end{table*}

\begin{figure}[p]
    \centering
    \includegraphics[width=0.7\textwidth]{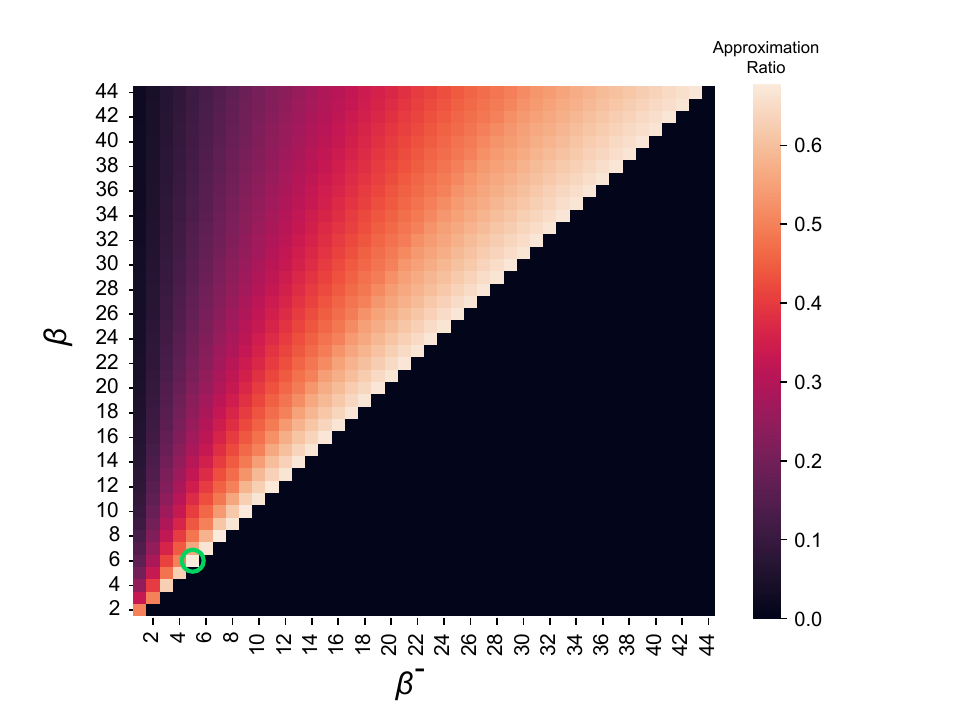}
    \caption{Heatmap shows the approximation ratio of $(\beta, \beta^{-})$-EDCS for different values of $\beta^{-}$ (x-axis) and $\beta$ (y-axis). The best approximation ratio is shown in a green circle for $\beta = 6$ and $\beta^{-} = 5$.}
    \label{fig:heatmap}
\end{figure}

The result of the factor-revealing LP shows that for $(\beta, \beta-1)$-EDCS when $\beta$ is an even number larger than 4, the approximation ratio of the EDCS is larger than 2/3 and as $\beta$ grows, the approximation ratio converges to 2/3 (see \Cref{fig:beta-minus-one}).

\begin{figure}[p]
    \centering
    \includegraphics[width=0.7\textwidth]{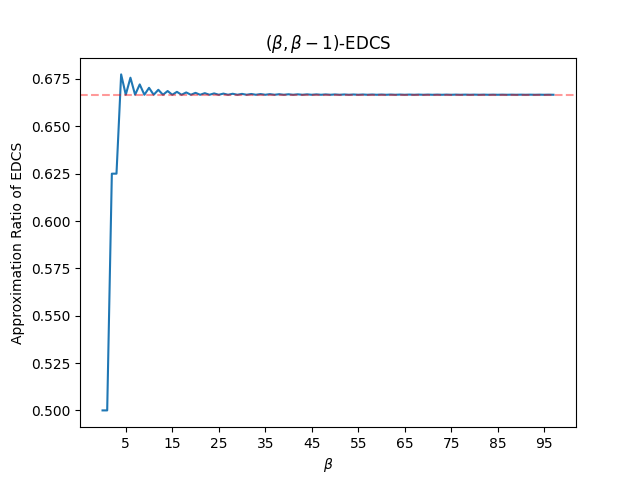}
    \caption{The approximation ratio of $(\beta, \beta-1)$-EDCS is computed for all possible values of $\beta \leq 100$. The y-axis denotes the approximation ratio, while the x-axis corresponds to the values of $\beta$. The horizontal red dashed line represents the approximation ratio $2/3$, which was previously believed to be the best possible approximation ratio of EDCS. The largest approximation ratio occurs when $\beta = 6$.}
    \label{fig:beta-minus-one}
\end{figure}

\begin{table}[p]
  \centering
  \begin{tabular}{|c||c|c|}
  \hline
     $\beta$ & $(\beta, \beta - 1)$-EDCS & $(\beta, \beta - 2)$-EDCS\\
    \hline
    5 & 0.6249 & 0.4782 \\ \hline

6 & 0.6774 & 0.5 \\ \hline

7 & 0.6666 & 0.5604 \\ \hline

8 & 0.6756 & 0.5783 \\ \hline

9 & 0.6666 & 0.6097 \\ \hline

10 & 0.6721 & 0.6153 \\ \hline

20 & 0.6678 & 0.6428 \\ \hline

30 & 0.6671 & 0.6511 \\ \hline

40 & 0.6669 & 0.6551 \\ \hline

50 & 0.6668 & 0.6575 \\ \hline

60 & 0.6667 & 0.659 \\ \hline

70 & 0.6667 & 0.6601 \\ \hline

80 & 0.6667 & 0.661 \\ \hline

90 & 0.6667 & 0.6616 \\ \hline

100 & 0.6667 & 0.6621 \\ \hline
\hline

  \end{tabular}
  \caption{The approximation ratio achieved by EDCS using the factor-revealing LP in \Cref{sec:lp}. Rows correspond to $\beta$, and columns correspond to instances of the EDCS that we are using (either $(\beta, \beta -1)$-EDCS or $(\beta, \beta -2)$-EDCS). The table presents the approximation ratio for some of the values of $\beta \leq 100$. $(\beta, \beta^{-})$-EDCS when $1 \leq \beta^{-} < \beta < 13$. As discussed before, for $(\beta, \beta - 1)$-EDCS, the approximation ratio is larger than 2/3 when $\beta$ is even and it converges to 2/3. On the other hand, the approximation ratio of $(\beta, \beta-2)$-EDCS is always smaller than 2/3 and converges to 2/3 as $\beta$ grows.}
  \label{tab:beta-1-beta-2}
\end{table}

On the other hand, for all other values of $(\beta, \beta^{-})$, the approximation ratio is always below 2/3, and for a constant integer $c \geq 2$, the approximation ratio of $(\beta, \beta - c)$-EDCS converge to 2/3 as $\beta$ goes to infinity (see \Cref{fig:different-betas} and \Cref{tab:beta-1-beta-2}).

\begin{figure}[p]
    \centering
    \includegraphics[width=0.65\textwidth]{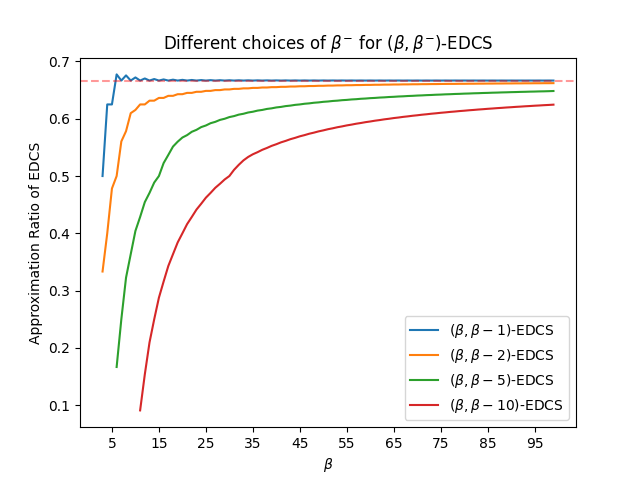}
    \caption{The approximation ratio of $(\beta, \beta-c)$-EDCS is calculated for various choices of $c \in [1, 2, 5, 10]$ across all feasible values of $\beta \leq 100$. The y-axis represents the approximation ratio, and the x-axis corresponds to the values of $\beta$.  The horizontal red dashed line represents the approximation ratio $2/3$. As $c$ increases, the approximation ratio worsens, and the convergence to 2/3 occurs at a slower rate.}
    \label{fig:different-betas}
\end{figure}

\section{Conclusion}

Over the recent years, EDCS has proven to be a successful matching sparsifier in various applications, including MPC model, stochastic matching model, sublinear time model, dynamic model, and streaming model. Many state-of-the-art results have been achieved by employing EDCS as a matching sparsifier in these settings.

It is well-known that for large values of $\beta$ and $\beta^{-}$, the approximation ratio of EDCS is $2/3-\epsilon$. This paper provides a tight analysis of the approximation ratio of $(\beta, \beta^{-})$-EDCS for small values of $\beta$ and $\beta^-$ using a factor-revealing linear program. Remarkably, we discover that when $\beta$ is even, the $(\beta, \beta-1)$-EDCS has an approximation ratio greater than 2/3, a previously unknown result. Our findings reveal that the maximum achievable approximation ratio is 0.6774 when $\beta = 6$ and $\beta^- = 5$. We hope the discovery that a (6, 5)-EDCS surpasses the 2/3 approximation, combined with the known bound that EDCS can be obtained through a simple greedy algorithm, opens avenues for future advancements in solving the maximum matching problem in different settings.

\bibliographystyle{plainnat}
\bibliography{references}
	
\end{document}